\documentclass[11pt]{article}

\usepackage{amsfonts}
\usepackage{amsmath}
\usepackage{enumitem}
\usepackage{mathtools}
\usepackage{centernot}
\usepackage{soul}
\usepackage{booktabs}
\usepackage{latexsym}
\usepackage[ruled]{algorithm2e}
\SetAlFnt{\small}
\SetAlCapFnt{\small}
\SetAlCapNameFnt{\small}
\SetAlCapHSkip{0pt}
\IncMargin{-\parindent}
\usepackage{algorithmic}
\usepackage{nicefrac}

\usepackage[
  bookmarks=true,
  bookmarksnumbered=true,
  bookmarksopen=true,
  pdfborder={0 0 0},
  breaklinks=true,
  colorlinks=true,
  linkcolor=black,
  citecolor=black,
  filecolor=black,
  urlcolor=black,
]{hyperref}
\usepackage{graphicx}
\usepackage{xcolor}
\usepackage[margin=1in]{geometry}
\usepackage{amssymb}
\usepackage{amsmath}
\usepackage[square]{natbib}
\setcitestyle{numbers}
\usepackage{amsthm}
\usepackage[capitalize]{cleveref}
\usepackage{bbm}
\usepackage{MnSymbol}

\Crefname{equation}{Equation}{Equations}

\theoremstyle{plain}
\newtheorem{theorem}{Theorem}
\newtheorem{corollary}[theorem]{Corollary}
\newtheorem{condition}[theorem]{Condition}
\newtheorem{lemma}[theorem]{Lemma}
\newtheorem{proposition}[theorem]{Proposition}
\newtheorem{remark}[theorem]{Remark}

\theoremstyle{definition}
\newtheorem{definition}[theorem]{Definition}
\newtheorem{notation}[theorem]{Notation}
\newtheorem{example}[theorem]{Example}

\DeclarePairedDelimiter\parentheses{(}{)}
\DeclarePairedDelimiter\braces{\{}{\}}
\DeclarePairedDelimiter\brackets{[}{]}
\DeclarePairedDelimiter\absolute{|}{|}
\DeclarePairedDelimiter\brackhalf{[}{)}

\begin{document}
\newcommand{\supp}{{\mathrm{supp}}}
\newcommand{\poly}{{\mathrm{poly}}}

\title{Multi-Channel Bayesian Persuasion}
\date{March 15, 2022}
\author{Yakov Babichenko\thanks{Technion--Israel Institute of Technology | \emph{E-mail}: \href{mailto:yakovbab@technion.ac.il}{yakovbab@technion.ac.il}.} 
\and 
Inbal Talgam-Cohen\thanks{Technion--Israel Institute of Technology | \emph{E-mail}: \href{mailto:italgam@cs.technion.ac.il}{italgam@cs.technion.ac.il}.} 
\and 
Haifeng Xu \thanks{University of Virginia | \emph{E-mail}: \href{mailto:hx4ad@virginia.edu}{hx4ad@virginia.edu}.} \and Konstantin Zabarnyi\thanks{Technion--Israel Institute of Technology | \emph{E-mail}: \href{mailto:konstzab@gmail.com}{konstzab@gmail.com}.}}
\maketitle

\begin{abstract}
The celebrated Bayesian persuasion model considers strategic communication between an informed agent (the sender) and uninformed decision makers (the receivers). The current rapidly-growing literature mostly assumes a dichotomy: either the sender is powerful enough to communicate separately with each receiver (a.k.a.~\emph{private} persuasion), or she cannot communicate separately at all (a.k.a.~\emph{public} persuasion). We study a model that smoothly interpolates between the two, by considering a natural multi-channel communication structure in which each receiver observes a subset of the sender's communication channels. This captures, e.g.,~receivers on a network, where information spillover is almost inevitable. 

We completely characterize when one communication structure is better for the sender than another, in the sense of yielding higher optimal expected utility universally over all prior distributions and utility functions. The characterization is based on a simple pairwise relation among receivers -- one receiver \emph{information-dominates} another if he observes at least the same channels. We prove that a communication structure $M_1$ is (weakly) better than $M_2$ if and only if every information-dominating pair of receivers in $M_1$ is also such in $M_2$. We also provide an additive FPTAS for the optimal sender's signaling scheme when the number of states is constant and the graph of information-dominating pairs is a directed forest. Finally, we prove that finding an optimal signaling scheme under multi-channel persuasion is, generally, computationally harder than under both public and private persuasion.
\end{abstract}

\section{Introduction}
\label{sec:intro}
Bayesian persuasion is the study of strategic information revelation by a \emph{sender} who wishes to influence \emph{receivers} to act in her favour~\cite{kamenica2011bayesian}. The sender's signals reveal to the receivers partial information about the state of the world, based on which they choose their actions and generate utility for the sender. The signals are sent according to a carefully-designed revelation policy called a \emph{signaling scheme}. Introduced in 2011, the Bayesian persuasion model has already raised fundamental algorithmic questions, has led to an explosive volume of research, and has had far-reaching economic impact (see, e.g., the surveys of \citet{dughmi2017algorithmic}, \citet{kamenica2019bayesian}, \citet{bergemann2019information}, and~\citet{Candogan20}).

Thus far, the persuasion of multiple receivers has mainly been studied in two dichotomous settings: \emph{public} persuasion, in which the sender announces a signal observed by all receivers~\cite[e.g.,][]{cheng2015mixture,bhaskar2016hardness,dughmi2019hardness, Ozan19};
and \emph{private} persuasion, in which the sender communicates with each receiver via a separate channel~\cite[e.g.,][]{babichenko2016computational,arieli2019private}. Several works compare between these two extremes~\cite{dughmi2017algorithmic, dughmi2017algorithmic2,bergemann2019information,taneva2019information}. However, in many prominent Bayesian persuasion applications, such as voting~\cite{alonso2016persuading,castiglioni2020persuading,castiglioni2021persuading,kerman2020persuading,kerman2021persuading} or advertising~\cite{fu2012ad,Miltersen12,Emek12,badanidiyuru2018targeting}, the sender-receivers communication takes a more intricate form. For example, a local politician may communicate with the public through several media channels, of which each individual can observe only a \emph{subset}. 
Similarly, a retailer may promote its brand by placing distinct advertisements on multiple online platforms, with different user populations exposed to different subsets of these ads.

\vspace{2mm}
\noindent{\bf Research goals: General communication structures.} In this paper, we consider a model of \emph{multi-channel Bayesian persuasion}, which captures generalized persuasion settings in-between public and private, by allowing arbitrary communication channels between the sender and subpopulations of receivers. Our first goal is to understand how the \emph{communication structure} -- that is, which information channels are observed by which receivers -- affects the persuasion problem of designing an optimal signaling scheme. Concretely, we pose the following question: {\bf When is one communication structure better for the sender than another?} To facilitate the comparison, we seek universally better structures -- that is, across all possible prior distributions over the states of nature and independently of sender's and receivers' utility functions.

This question is closely related to an intriguing result of~\citet{galperti2019belief}, analogous to our Theorem~\ref{thm:order}; our result was obtained independently and later than their (see also a latter version~\citep{galperti2020belief}). They study a setting with agents located in a directed network that specifies the information spillover, and the sender provides some information to a subset of her choice of these agents -- the subset of the \emph{seeds}. All agents in the network transmit all the information they receive to their out-neighbors. After the communication phase, agents are playing a game with externalities. This model can be viewed as multi-channel Bayesian persuasion, where the channels are the seeds, and agent $i$ observes channel $j$ iff there is a path from seed $j$ to agent $i$ in the network. We use very similar ideas to characterize the feasible outcomes. In particular, the key idea of utilizing the secrete sharing protocol in Bayesian persuasion, and the notion of \emph{information dominating pairs} -- referred to as \emph{more connected agents} in~\cite{galperti2019belief} -- play a central role in both works.

Our second goal is to begin the exploration of the new computational landscape arising from the introduction of communication structures into the standard persuasion model. While, in general, we show that this could introduce computational hardness to previously tractable settings, we seek classes of structures that allow efficient algorithms for approximately-optimal persuasion.

\vspace{2mm}
\noindent{\bf A characterization result.}
The following simple notion plays a central role in the characterization result. Given a communication structure $M$, we say that an ordered pair of receivers $\parentheses*{i_1,i_2}$ is \emph{information-dominating} if receiver $i_1$ observes all the communication channels that receiver $i_2$ does. One implication of $\parentheses{i_1,i_2}$ being information-dominating is that for every signaling scheme of the sender, receiver $i_1$ is weakly more informed than receiver $i_2$ in the \emph{Blackwell sense}~\cite{blackwell1950comparison, Blackwell53}.\footnote{That is, the distribution of posterior beliefs of $i_1$ regarding the state of nature under any signaling scheme must be a \emph{mean-preserving spread} of the distribution of $i_2$ -- see Subsection~\ref{sub:FPTAS}.}
We use the information-domination primitive to fully answer the question above -- i.e.,~to characterize when one communication structure is preferable to another in arguably the strongest possible sense.

A first observation is that if $\parentheses*{i_1,i_2}$ is information-dominating in a communication structure $M_1$ but not in another structure $M_2$, then it cannot be the case that $M_1$ is always better than $M_2$ for the sender. Intuitively, an information-dominating pair of receivers restricts the extent upon which the sender can use signals to anti-correlate their posterior beliefs; one can apply this intuition to devise the sender's and receivers' utilities such that this restriction becomes a strict disadvantage to the sender. Hence, if $M_1$ is always weakly preferable over $M_2$ for the sender, then the set of all information-dominating pairs of $M_1$ must be contained in that of $M_2$. 

More surprisingly, we show that the opposite direction of the above statement is also true -- if the information-dominating pairs of $M_1$ are also such in $M_2$, then $M_1$ is weakly better than $M_2$ for the sender (Theorem~\ref{thm:order}). This result has the following interpretation -- apart from the information-dominating pairs of receivers, which impose restrictions on the receivers' admissible beliefs, no additional restrictions are imposed. Our characterization result holds in the most general model of Bayesian persuasion in which receivers may have \emph{externalities},\footnote{The presence of \emph{externalities} means that the receivers' actions affect each other; mathematically, after observing the signal realizations -- the receivers play a Bayesian game with incomplete information.} even though it is well-known to be notoriously difficult to characterize optimal signaling schemes in the presence of such externalities~\cite{taneva2019information}.

As a corollary of our characterization result, the sender can implement private persuasion in every communication structure with no information-dominating pairs (Corollary~\ref{cor:optimal}). Intriguingly, this implies that private persuasion for $k$ receivers can be implemented using only $O\parentheses*{\log k}$ communication channels (Proposition~\ref{pro:private}). Note that private signaling is an optimal communication structure from the sender's point of view, as in this structure she can emulate any signaling scheme implementable in any other structure (Corollary~\ref{cor:best}).

\vspace{2mm}
\noindent{\bf A positive computational result.} In Section~\ref{sec:alg}, we demonstrate the new algorithmic questions that arise from our multi-channel model. In Subsection~\ref{sub:FPTAS}, we study a fundamental special case of multi-channel persuasion in which the directed graph representing the information-dominating receiver pairs is a \emph{forest}. Communication structures corresponding to forests capture natural communication patterns in hierarchical, compartmentalized organizations.\footnote{Consider, e.g.,~an intelligence agency, the head of which (sender) communicates information to her subordinates (receivers) through several separate communication channels; the information-domination tree represents the organizational hierarchy.}

For this case of induced forests, we provide in Theorem~\ref{thm:FPTAS} an additive FPTAS for the optimal signaling scheme, under the assumptions of a constant number of states of nature, no externalities among receivers and a sender with an additive utility function. The latter two restrictions are commonly adopted in recent algorithmic studies \cite{dughmi2017algorithmic2,hahn2020prophet}, whereas the first assumption of a constant number of states is necessary for circumventing known hardness results; in particular, with a non-constant number of states, even public persuasion is known to be hard.\footnote{\citet{dughmi2017algorithmic2} prove that one can achieve neither an additive PTAS nor a constant-factor approximation of the optimal sender's utility in time polynomial in the number of states, unless P $=$ NP and even if the sender's utility is additive.} An interesting direction for future research is the general relation between the graph-theoretic properties induced by the communication structure and the computational tractability of the optimal signaling scheme.

\vspace{2mm}
\noindent{\bf Hardness of multi-channel persuasion.} Multi-channel persuasion is, generally, harder than public or private persuasion -- both conceptually and computationally. In Subsection~\ref{sub:hardness}, we prove that for a simple and general family of sender's utility functions, \emph{separable supermajority} functions, finding an optimal signaling scheme in general multi-channel persuasion is NP-hard, while it is computationally tractable for both public and private persuasion (Theorem~\ref{thm:hardness}).
\subsection{Related Literature}

\vspace{2mm}
\noindent{\bf Between public and private persuasion.}
As mentioned above, the vast majority of the algorithmic and economic literature on Bayesian persuasion with multiple receivers focuses on either private or public persuasion, often comparing and contrasting these two extremes. We describe here works that diverge from this pattern and are, thus, related to our model of multi-channel persuasion. We have already mentioned~\citet{galperti2019belief,galperti2020belief} as the most closely related work to ours. \citet{kerman2021persuading}~consider~\emph{communicating voters}: the voters (receivers) are located on a network; a politician (sender) tries to persuade them to vote for a proposal by sending them private signals; each voter ends up observing his own signal \emph{and} the signals of his neighbours. Indeed, on social media, when a voter targeted by a politician shares the politician's message, his friends also see the message. This model corresponds to voters who communicate with each other by a single communication phase, in which every voter reveals his private signal to his~neighbours before deciding how to vote.  Our characterization result has interesting implications to such a model. In particular, it immediately yields some results of~\cite{kerman2021persuading} (see Subsection~\ref{sub:net}). \citet{CandoganSpillovers20}~studies information spillover among receivers on a network in a general (not necessarily voting) context. He restricts attention to a special class of signaling schemes -- \emph{threshold schemes} -- and to a specific family of communication structures in which the receivers arrive sequentially, and each receiver see all his ancestors' information. This work presents a ``sequential'' view in which receivers see the realized signals from their predecessors in an acyclic network.

\citet{brooks2019information} study a classic question in economics regarding \emph{information orders} -- what is the relation between being ``more informed'' about the state of nature in the Blackwell sense, and being ``more informed'' in the sense of observing a refined signal about the state? Their model considers an \emph{information hierarchy}, which they define as a directed graph with the receivers as the vertices and the Blackwell-ordered receiver pairs as edges. This work is connected to ours as any communication structure defines a (possibly intricate) information hierarchy.\footnote{A setting with $n$ communication channels and $k=2^n$ receivers in which every agent observes a different subset of channels is an example of an intricate information hierarchy.} Unlike~\cite{brooks2019information}, who study generic information hierarchies, in our case the sender's signaling scheme specifies the information hierarchy. Interestingly, the results of~\cite{brooks2019information} depend on the information hierarchy, and their positive results only hold for \emph{forest hierarchies}; in contrast, in our problem, the naive consideration of information-dominating pairs suffices to determine what utility the sender can achieve and how. While~\cite{brooks2019information} do not consider algorithmic aspects, their results for forest hierarchies prove useful for our algorithmic result in Subsection~\ref{sub:FPTAS}. 

\citet{mathevet2020organized} study how to optimally transmit information in organizations, exploring as a contrast to public signaling the notion of \emph{vertical} signaling -- that is, using communication over the organizational hierarchy. A crucial focus of their model, unlike ours, is that receivers \emph{strategically} choose what to pass on down the hierarchy.

\vspace{2mm}
\noindent
{\bf
Externalities among receivers.} 
Many previous works on Bayesian persuasion with multiple receivers assume no externalities among receivers~\cite[e.g.,][]{babichenko2016computational,arieli2019private,xu2020tractability}.
With externalities, hardness results have been established even for two receivers and even for public persuasion~\cite{bhaskar2016hardness,Rubinstein17,dughmi2019hardness}. However, the general case in which receivers play a Bayesian game with externalities has been considered in economics~\cite{bergemann2019information,taneva2019information,forges2020games,mathevet2020information,koessler2021splitting}. The classic work on persuasion in Bayesian games with incomplete information by~\citet{bergemann2013robust, bergemann2016bayes} provides a powerful, simplified approach for dealing with multi-receiver persuasion with externalities through the notion of \emph{Bayesian correlated equilibrium}. Unfortunately, one cannot apply this approach to our multi-channel persuasion model, as it requires \emph{private} recommendations to each receiver. 

\vspace{2mm}
\noindent
{\bf
Cryptographic-inspired proofs.}
The proof of our characterization result (Theorem~\ref{thm:order}) uses cryptographic-flavoured techniques. This proof is similar in spirit to secret sharing protocols, which are commonly studied in cryptography~\cite{blakley1979safeguarding, shamir1979share}. Such protocols allow to reveal a certain secret to any set of agents containing at least one subset from a pre-specified collection $C$, by sending private signals to the agents, such that no set of agents that does not contain any element of $C$ can learn anything about the secret. In our paper, each receiver can be considered a set of channels - ``agents'' to which the sender aims to reveal a shared secret. We simultaneously share different ``secrets'' between different sets of channels.

To our knowledge, proofs similar to the one of Theorem~\ref{thm:order} are not common in the study of Bayesian persuasion, besides~\cite{galperti2019belief, galperti2020belief}, who use them in a similar manner to ours. Two related examples from other subfields of game theory are:
\begin{enumerate}
    \item Implementing correlated equilibria via cryptography~\cite{dodis2000cryptographic}.
    \item Analyzing competing-mechanism games~\cite{attar2021keeping}.
\end{enumerate}
In the latter example, private disclosures are shown to generate equilibrium outcomes that cannot be supported otherwise, since these disclosures can contain encryption keys.
\section{Preliminaries}
\label{sec:prelims}
\noindent
{\bf Standard notations.} Let $\brackets*{c}:=\braces*{1,\ldots,c}$ for every positive integer $c$. For a set $B$, denote by $\Delta\parentheses*{B}$ the set of all the probability measures over $B$.

\vspace{2mm}
In this section, we describe our persuasion model for the characterization result in Section~\ref{sec:characterization}, which generalizes standard algorithmic Bayesian persuasion in two ways:
\begin{enumerate}
    \item It encompasses public and private persuasion, as well as all settings ``in between'' public and private, by introducing the simple notion of a communication structure.
    \item Unlike most algorithmic works on persuasion, it allows externalities among the receivers -- i.e., a receiver's utility is allowed to depend on others' beliefs and actions -- through Harsanyi's universal type space; this makes our results robust to receivers' externalities. For simplicity of presentation, we relegate the description of the general case through universal type spaces to Appendix~\ref{app:univarsal}.
\end{enumerate}

\vspace{2mm}
\noindent
{\bf Model overview.} We consider a single-sender, multi-receiver Bayesian persuasion environment with $k$~receivers $R_1,\ldots,R_k$.\footnote{We address the sender as female and the receivers as males.} There are $n$ communication channels $C_1,...,C_n$.

\vspace{2mm}
\noindent
{\bf Communication structure.}
Our model diverges from most previous work on persuasion in that the sender can transmit signals separately to (possibly overlapping) subsets of receivers, as specified by a communication structure. Consider a $k\times n$ binary matrix $M$ where $M_{i,j}=1$ if receiver $R_i$ observes the signal that is sent via the communication channel $C_j$ and $M_{i,j}=0$ otherwise. Call the matrix $M$ the \emph{communication structure}.

Denote the state of nature space by $\Omega$, and the common prior by $p\in\Delta\parentheses*{\Omega}$. For simplicity of exposition, we set $\Omega=\brackets*{0,1}$, but all our results apply also in the more general case with $\Omega\subseteq\brackets*{0,1}^d$ for some positive integer $d$.

\vspace{2mm}
\noindent
{\bf Signaling schemes.} A \emph{signalling scheme} of the sender consists of $n$ sets of possible \emph{signal realizations} $S_1,...,S_n$ and a mapping $\pi:\Omega \to \Delta \parentheses*{\times_{j\in \brackets*{n}} S_j}$; that is, $\pi$ assigns a distribution over the combination of signals in each state. Denote the random variable representing the combination of signals by $\sigma=\parentheses*{\sigma_j}_{j\in \brackets*{n}}$.

We restrict attention to signaling schemes with $\absolute*{S_j}\leq \aleph$ for $j\in\brackets*{n}$; that is, at most $\aleph$ different signals are allowed to be sent on each communication channel.\footnote{Recall that $\aleph$ denotes the continuum cardinality. More specifically, we shall assume w.l.o.g.~that the set of signals is a subset of $\brackhalf*{0,1}$.} In Appendix~\ref{app:aleph}, we show that this assumption is not restrictive; namely, we prove that $\aleph$ signals suffice for optimal persuasion in our~setting.

After observing the signal combination $\parentheses*{\sigma_j}_{j:M_{i,j}=1}$, the receiver $R_i$ forms a (marginal) \emph{posterior belief} about the state, $p_{\sigma}^i \in \Delta\parentheses*{\Omega}$. Formally, for every measurable set $B\subseteq \Omega$ and $i\in\brackets*{k}$, we have:

$$p_{\sigma=\parentheses*{s_1,\ldots,s_n}}^i\parentheses*{B}=Pr_{p,\pi}\brackets*{B|\sigma_j=s_j\;\;\forall j:\; M_{i,j}=1}.$$

We denote the \emph{joint posterior distribution} by $p_{\sigma}=\parentheses*{p_{\sigma}^1,\ldots,p_{\sigma}^k}\in\parentheses*{\Delta\parentheses*{\Omega}}^k$. Every signaling scheme $\pi$ induces a distribution over the product of states and profiles of posterior beliefs $P=P_{\pi}\in \Delta\parentheses*{\Omega \times \parentheses*{\Delta\parentheses*{\Omega}}^k}$ as follows. First, the state is drawn according to $\mu$. Thereafter, the signal combinations are drawn according to $\pi$. Finally, each receiver observes the signals in the corresponding communication channels and forms a posterior distribution.
  
\vspace{2mm}
\noindent
{\bf Utilities, the game and the equilibrium selection rule.} After observing the signals, the receivers are playing some \emph{continuation game with incomplete information $G$} (possibly with externalities). We fix compact and not-empty \emph{action spaces} $A_1,...,A_k\subseteq\mathbb{R}$ of $R_1,...,R_k$, respectively, and denote $A:\times_{i\in \brackets*{k}} A_i$. Let $\tau_i\in\Delta\parentheses*{A_i}$ ($i\in\brackets*{k}$) be the mixed strategy of the receiver $R_i$ in $G$ and $\tau=\times_{i\in \brackets*{k}} \tau_i$ the strategy profile.

Let $v:\Omega \times A \to \mathbb{R}$ be the sender's utility function, and let, for every $1\leq i\leq k$, $u_i\parentheses*{\omega,a_1,\ldots,a_k}:\Omega\times A\to\mathbb{R}$ be $R_i$'s utility function. To perform an optimization analysis from the sender's perspective, we should apply some \emph{equilibrium selection rule} $f$ on the way the game $G$ is played after each realization of signals. Formally, $f$ is a mapping from the space of all possible joint posteriors $\parentheses*{\Delta\parentheses*{\Omega}}^k$ to the set of the \emph{Bayes-Nash equilibria} in $G$. In a Bayes-Nash equilibrium, the strategy profile satisfies for every $i\in\brackets*{k}$, marginal posterior $p_{\sigma}^i\in \Delta\parentheses*{\Omega}$ and alternative strategy $\tau_i'\in\Delta\parentheses*{A_i}$ that $R_i$'s expected utility cannot be improved by a unilateral deviation to $\tau_i'$: 
\begin{align*}
&\mathbb{E}_{\omega'\sim p_{\sigma}^i, a_1'\sim\tau_1,\ldots,a_k'\sim\tau_k}\brackets*{u_i\parentheses*{\omega',a_1',\ldots,a_k'}}\geq\\
&\mathbb{E}_{\omega'\sim p_{\sigma}^i,\parentheses*{a_{i'}':i'\in\brackets*{k}\setminus\braces*{i}}\sim\parentheses*{\tau_{i'}:i'\in\brackets*{k}\setminus\braces*{i}},a_i''\sim\tau_{i}'}\brackets*{u_i\parentheses*{\omega',\parentheses*{a'_{i'}:i'\in\brackets*{k}\setminus\braces*{i}},a_i''}}.
\end{align*}

We impose a single requirement on the selection rule $f$: \emph{the outcome of the game is fully determined by $P$}, the distribution over states and profiles of beliefs. Formally, for every two signaling policies $\pi$ and $\pi'$ that induce the same distribution over posteriors $P$, the selected equilibria in the continuation games coincide. This requirement essentially means that signals do not play a role but only the posterior that they induce. This requirement is (implicitly) adopted in \emph{all} literature on Bayesian persuasion we are familiar with.\footnote{Yet, it is not hard to come up with selection rules violating this requirement.} 

Given a communication structure $M$, we denote by $\mathcal{P}_M$ the set of distributions $P$ that are implementable by some signaling policy $\pi$. Every such $P$ induces an (expected) sender's utility of $u\parentheses*{P}:=\mathbb{E}_{\parentheses*{\omega',p_{\sigma}}\sim P} \brackets*{v\parentheses*{\omega',f\parentheses*{p_{\sigma}}}}$. Note that $u$ -- the \emph{indirect sender's utility}~\cite{kamenica2011bayesian} -- is a mapping from $\Delta\parentheses*{\Omega \times \parentheses*{\Delta\parentheses*{\Omega}}^k}$ to $\mathbb{R}$.
\begin{remark}
\label{rem:general}
In general, the action of each receiver may be affected not only by her own belief on the distribution of $\omega$, but also by her belief on the beliefs of the other receivers on the distribution of $\omega$, and even by higher-order beliefs -- that is, the receivers may have \emph{externalities}. In this case, the distribution over posteriors $P$ as we defined it may be insufficient to determine the sender's utility, as the posterior should capture also high-order beliefs, but our results continue to hold -- please see Appendix~\ref{app:univarsal} for details.
\end{remark}
The sender's optimization target can be formulated as $V\parentheses*{M,u}=\sup_{P\in \mathcal{P}_M} u\parentheses*{P}$.\footnote{Since we did not impose an upper-semicontinuity or an analogous requirement on the equilibrium selection rule, the supremum might not be attainable.} The following summarizes the persuasion process:
\begin{enumerate}[itemsep=0mm]
    \item Knowing the communication structure $M$ and the Bayesian persuasion instance, the sender commits to a signaling scheme $\pi$.
    \item The sender observes the true state $\omega\in \Omega$ drawn according to the prior.
    \item The sender transmits signals $\sigma_1,\ldots,\sigma_n$ according to $\pi$ and $\omega$.
    \item For every $i\in\brackets*{k}$ and $j\in\brackets*{n}$, receiver $R_i$ observes signal $\sigma_j$ if and only if $M_{i,j}=1$; he forms a posterior belief $p_{\sigma}^i\in \Delta\parentheses*{\Omega}$ about $\omega$ according to the signal realizations he observes.
    \item For the joint posterior $p_{\sigma}=\parentheses*{p_{\sigma}^1,\ldots,p_{\sigma}^k}$ and a corresponding Bayes-Nash equilibrium $\parentheses*{\tau_1,\ldots,\tau_k}=f\parentheses*{p_{\sigma}}$, each receiver $R_i$ takes an action $a_i\in A_i$, where $a_i\sim \tau_i$ (independently of the other receivers' actions).
    \item The sender gets utility of $u\parentheses*{\omega,a_1,\ldots,a_k}$, while receiver $R_i$ gets utility of $u_i\parentheses*{\omega, a_1,\ldots,a_k}$ ($i\in\brackets*{k}$).
\end{enumerate}
To facilitate the formulation of our robust characterization result in the next section, we shall define the following.
\begin{definition}
\label{def:setting}
Fixing a communication structure $M$, a Bayesian persuasion \emph{instance} is given by a tuple $\parentheses*{\Omega,p,A,u,u_1,\ldots,u_k,f}$ of state space, prior distribution, action space, utilities and equilibrium selection rule. 
\end{definition}
\section{A Hierarchy of Communication Structures}
\label{sec:characterization}
In this section, we compare communication structures \emph{robustly with respect to the persuasion instance}.
\begin{definition}[Partial order on communication structures]
\label{def:superior}
We say that a communication structure $M$ is \emph{superior} over a communication structure $M'$ if $V\parentheses*{M,u}\geq V\parentheses*{M',u}$ for \emph{every} Bayesian persuasion instance. In such a case, we denote $M'\preceq M$.
\end{definition}
Note that $\preceq$ is a a partial order on the space $\braces*{0,1}^{k\times n}$ of communication structures.
\begin{remark}
\label{rem:induced}
It is not clear whether all utilities $u:\Delta\parentheses*{\Omega \times \parentheses*{\Delta\parentheses*{\Omega}}^k}\to \mathbb{R}$ are, indeed, induced by some continuation game with incomplete information $G$ and some equilibrium selection rule $f$ (unless the game is trivial and all action profiles are equilibria). Whenever we shall prove superiority of one communication structure compared to another, the results will apply for arbitrary $u$. Whenever we shall prove absence of superiority, we shall explicitly construct a simple continuation game, even without extrenalities, that violates superiority.
\end{remark}
We first present our characterization result (Theorem~\ref{thm:order}), which fully characterizes the partial ordering $\preceq$ of communication structure by superiority; then we provide some motivating implications in Subsection~\ref{sub:implications} and a proof overview in Subsection~\ref{sub:characterization-pf}; Subsection~\ref{sub:net} discusses the application of our characterization result to receivers organized in a network. The following definition plays a central role in our result.
\begin{definition}
An ordered pair $\parentheses*{i_1,i_2}\in \brackets*{k}\times \brackets*{k}$ ($i_1\neq i_2$) is \emph{information-dominating} under a communication structure $M$ if $M_{i_1,j}\geq M_{i_2,j}$ for every channel $j\in \brackets*{n}$. In this case, we say that receiver $R_{i_1}$ \emph{information-dominates} receiver $R_{i_2}$.
\end{definition}
That is, $R_{i_1}$ information-dominates $R_{i_2}$ if $R_{i_1}$ observes all communication channels that $R_{i_2}$ does (and possibly more).
\begin{notation}
\emph{$S_M\subseteq \brackets*{k}\times \brackets*{k}$} is the set of all the information-dominating pairs under the communication structure $M$.
\end{notation}
The following theorem fully characterizes $\preceq$ via the sets of information-dominating pairs induced by every communication structure.
\begin{theorem}[Similar to~\citet{galperti2019belief}]
\label{thm:order}
For every two communication structures  $M_1,M_2\in \braces*{0,1}^{k\times n}$, $M_1\preceq M_2$ if and only if $S_{M_1}\supseteq S_{M_2}$.
\end{theorem}
Moreover, the proof of Theorem~\ref{thm:order} immediately yields the following useful property:
\begin{proposition}
\label{pro:implement}
$M_1\preceq M_2$ if and only if $\mathcal{P}_{M_1}\subseteq\mathcal{P}_{M_2}$.
\end{proposition}
That is, a communication structure $M_2$ is superior over $M_1$ iff every distribution over posteriors implementable under $M_1$ is also implementable under $M_2$.
\subsection{Motivating Implications}
\label{sub:implications}
We present several immediate implications of Theorem~\ref{thm:order}. 
The first is a straightforward exercise of the definitions.
\begin{corollary}[Private persuasion is the optimal communication structure]
\label{cor:best}
Let $M^*$ be a communication structure corresponding to private persuasion.\footnote{E.g., $M^*_{i,j}=1$ if and only if $i,j\in\brackets*{k}$ and $i=j$.} Then for every communication structure $M\in \braces*{0,1}^{k\times n}$ it holds that $M\preceq M^*$. That is, private persuasion is optimal for the sender in any persuasion instance.
\end{corollary}
Indeed, for $M^*$ we have $S_{M^*}=\emptyset$; hence,  $S_{M^*} \subseteq S_M$ for every $M$, which implies by Theorem~\ref{thm:order} that $M \preceq M^*$. Corollary~\ref{cor:best} can alternatively be observed through Proposition~\ref{pro:implement}, since any information sent via some communication structure can be sent privately to the receivers for whom it is intended.

More surprisingly, it turns out that \emph{every} communication structure in which no two receivers are comparable in terms of the information they receive is equivalent to private persuasion in terms of the sender's optimal utility; in particular, all such communication structures are optimal for the sender.
\begin{corollary}[Communication structures with no information-dominating pairs are optimal]
\label{cor:optimal}
Let $M^*$ be a communication structure corresponding to private persuasion.
If a communication structure $M$ has no information-dominating pairs, then $M^*\preceq M$. That is, $M$ is optimal for the sender.
\end{corollary}
In fact, by Proposition~\ref{pro:implement}, every distribution $P$ over states and posteriors induced by some signaling scheme under private persuasion can be induced by an appropriate signaling scheme under $M$, making such communication structures fully equivalent to private persuasion.

The next proposition shows that private persuasion for $k$ receivers can be implemented using only $O\parentheses*{\log k}$ communication channels. That is, the sender can ensure the same optimal expected utility (regardless of the persuasion instance) as she can by private persuasion after an \emph{exponential} reduction in the number of communication channels. Our construction matches the information-theoretical lower bound on the number of channels (not only asymptotically, but exactly).
\begin{proposition}
\label{pro:private}
For $k$ receivers, let $m=m\parentheses*{k}$ be the smallest integer such that ${m \choose \lfloor \frac{m}{2} \rfloor}\geq k$.\footnote{Note that by Stirling's approximation, $m=\Theta\parentheses*{\log k}$.} Then the communication structure with $n=m$ channels in which each receiver observes a different subset of size $\lfloor \frac{m}{2} \rfloor$ of the channels is equivalent to private persuasion. Moreover, $m\parentheses*{k}$ is the minimal number of channels sufficient to implement a communication structure equivalent to private persuasion for $k$ receivers.
\end{proposition}
\begin{proof}
The first part of the proposition immediately follows from Corollary~\ref{cor:optimal}, as the suggested communication structure has no information-dominating pairs. For the second part, we use Sperner's theorem~\cite{sperner1928satz}, which states that the maximum size of a family of subsets of $\brackets*{m}$ such that no subset contains another is ${m \choose \lfloor \frac{m}{2} \rfloor}$. Therefore, any communication structure using less than $m\parentheses*{k}$ channels has an information-dominating pair, and by Theorem~\ref{thm:order} -- is not equivalent to private persuasion.
\end{proof}
\subsection{Proof Overview of the Characterization Result}
\label{sub:characterization-pf}
We shall present here the key ideas of Theorem~\ref{thm:order} proof; for the full proof, please see Appendix~\ref{app:characterization}.
\begin{proof}[Proof idea of Theorem~\ref{thm:order}]
If $M_1\preceq M_2$, then it is somewhat expected that $S_{M_1}\supseteq S_{M_2}$, as Blackwell-comparability of the information the receivers obtain imposes severe limitations on the distributions over posteriors implementable by signaling schemes. Suppose, e.g., that the first receiver is always (weakly) more informed than the second one. Then, for appropriately chosen utility functions (even with no externalities) and given an optimal strategy of the second receiver, the sender cannot misalign this strategy with the first receiver's strategy.

The second direction is more surprising. Here we briefly explain why Corollary~\ref{cor:optimal} holds, and Theorem~\ref{thm:order} follows from an inductive generalization of the idea we describe. Suppose, e.g., that the sender wishes to transmit a private signal to receiver $R_1$ using a communication channel $C$ shared with $l$ more receivers. Then she can encrypt the signal with $l$ i.i.d.~uniform keys, sending each key via a communication channel observed by $R_1$, but not by the corresponding receiver among the $l$~additional ones (such a channel necessarily exists by the assumption in Corollary~\ref{cor:optimal}). Then $R_1$ will be able to decrypt his signal, while the other receivers will only see some random i.i.d.~values. This communication can be seen as applying a secret sharing protocol in which the ``agents'' who share the secret are the communication channels that $R_1$ observes. A careful application of the same idea allows us to handle the case of general setting instances, in which the receivers also play a Bayesian game with incomplete information. We overcome the technical challenges imposed by externalities using Harsanyi's universal type~spaces (see Appendix~\ref{app:univarsal} for details). 
\end{proof}
\subsection{Receivers on a Network}
\label{sub:net}
In this subsection, we discuss the implications of our model to the setting of~\citet{kerman2021persuading} with receivers on a network. In particular, some of their results that show that certain communication structures are equivalent to private persuasion for specific utility functions immediately follow from our work.

\vspace{2mm}
\noindent{\bf The setting.} Fix an undirected graph $G=\parentheses*{\brackets*{k}, E}$, the vertices of which represent the $k$ receivers and the edges of which represent social or other ties. Define $n=k$ communication channels and a communication structure $M_G$ such that ${(M_G)}_{i,j}=1$ ($i,j\in\brackets*{k}$) if and only if either $j=i$ or $\braces*{i,j}\in E$. That is, the $i$-th communication channel is observed by $R_i$ and by its neighbours in $G$. A possible motivation for this model is communication on social media.

\vspace{2mm}
\noindent{\bf Implications of Corollary~\ref{cor:optimal}.} It turns out that many natural classes of graphs $\{G\}$ yield communication structures $\{M_G\}$ with no information-dominating pairs. The intuitive reason is that no receiver's social ties are overshadowed in full by those of a friend. We exploit this fact
to immediately deduce several results of~\cite{kerman2021persuading}, the original proofs of which were technically involved.
\begin{condition}
\label{cond:private}
For every $\parentheses*{i_1,i_2}\in \brackets*{k}\times\brackets*{k}$ such that $\braces*{i_1,i_2}\in E$, there exists $i'\in \brackets*{k}\setminus \braces*{i_1,i_2}$ for which $\braces*{i_1,i'}\notin E$, but $\braces*{i_2,i'}\in E$.
\end{condition}
Whenever $G$ satisfies Condition~\ref{cond:private}, one can apply Corollary~\ref{cor:optimal} to deduce that $M_G$ is equivalent to private persuasion. Indeed, an information-dominating pair $\parentheses*{i_1,i_2}$ must satisfy that $\braces*{i_1,i_2}\in E$ (otherwise -- the channel observed by $R_{i_2}$ and his neighbours is not observed by $R_{i_1}$); however, the $i'$-th channel is observed by $R_{i_2}$, but not by $R_{i_1}$. We present two examples.
\begin{example}
\label{ex:circle}
Suppose that $G$ is a \emph{circle} -- a connected graph with all degrees equal to $2$ (Proposition~$5.4$ from~\cite{kerman2021persuading}). Provided that $k>3$, Condition~\ref{cond:private} holds and one can apply Corollary~\ref{cor:optimal}.
\end{example}
\begin{example}
\label{ex:grid}
Suppose that $G$ is obtained from a $2$-dimensional grid by removing any number of edges such that the degree of each vertex remains at least $2$ (a generalization of Example~$5.9$ from~\cite{kerman2021persuading}).\footnote{Formally, $G$ is a \emph{$2$-dimensional grid} if $k=k_1\times k_2$ for some positive integers $k_1, k_2$, each vertex of $G$ is labeled by another pair $\parentheses*{l_1,l_2}\in\brackets*{k_1}\times\brackets*{k_2}$ and $\braces*{\parentheses*{l_1,l_2},\parentheses*{l_1',l_2'}}\in E$ if and only if $\absolute*{l_1'-l_1}+\absolute*{l_2'-l_2}=1$.} Again, Condition~\ref{cond:private} holds and Corollary~\ref{cor:optimal} applies.
\end{example}
\section{Algorithmic Multi-Channel Persuasion}
\label{sec:alg}
In this section, we show that finding an approximately optimal signaling scheme is computationally tractable when the communication structure is represented by a directed forest (Subsection~\ref{sub:FPTAS}); however, it is not the case for general communication structures, even for simple sender's utility functions (Subsection~\ref{sub:hardness}). We assume for this section that the receivers have no externalities -- that is, their utility is affected only by their own action and true state of nature. Moreover, we focus on a constant number of states $\absolute*{\Omega}$.\footnote{One can consider it a special case of $\Omega=\brackets*{0,1}$, with the prior being supported only on finitely many points of $\Omega$.}
\subsection{Approximate Tractability of Forest Communication Structures}
\label{sub:FPTAS}
In this subsection, we focus on communication structures satisfying the following condition.
\begin{definition}
\label{def:forest}
Given a communication structure $M$, consider the directed graph $G^M$ with the vertex set being $\brackets*{k}$ and $\parentheses*{i_1,i_2}\in \brackets*{k}\times\brackets*{k}$ being an edge iff $\parentheses*{i_1,i_2}$ is an information-dominating pair and there is no $i_3\in\brackets*{k}\setminus\braces*{i_1,i_2}$ s.t.~both $\parentheses*{i_1,i_3}$ and $\parentheses*{i_3,i_2}$ are information-dominating pairs. We say the setting has a \emph{forest communication structure} if $G^M$ is a directed forest. 
\end{definition}
We consider \emph{additive} sender’s utility functions.  That is, we assume that there exist functions $u^1,\ldots,u^k:\Delta\parentheses*{\Omega}\to\mathbb{R}$ for which $u\parentheses*{p_{\sigma}}=\sum_{i=1}^{k} u^i\parentheses*{p_{\sigma}^i}$, where $p_{\sigma}^i\in\Delta\parentheses*{\Omega}$ is the marginal posterior distribution induced on $R_i$ by $p_{\sigma}$. We further assume that the sender has an oracle access to $u^1,\ldots,u^k$. Moreover, we focus on either $O\parentheses*{1}$-Lipschitz or piecewise constant functions $u^1,\ldots,u^k$. The former is a standard assumption for receivers with continuum-sized action spaces, while the latter captures the case of finite-action-space receivers~\cite{dughmi2017algorithmic}. Under these conditions, we provide an additive FPTAS for the optimal signaling scheme.
\begin{theorem}
\label{thm:FPTAS}
Consider a forest communication structure $M$, constant $\absolute*{\Omega}$ and sender's utility function given by $u\parentheses*{p_{\sigma}}=\sum_{i=1}^{k} u^i\parentheses*{p_{\sigma}^i}$ for some $u^1,\ldots,u^k:\Delta\parentheses*{\Omega}\to\mathbb{R}$, such that each one of them is either $O\parentheses*{1}$-Lipschitz or piecewise constant. Then there exists an algorithm outputting, for every $\epsilon>0$, an $\epsilon$-additively optimal signaling scheme in $\poly\parentheses*{k,\frac{1}{\epsilon}}$-time.\footnote{Specifically, our algorithm outputs a table specifying, for each state of nature, which signal should be sent with which probability.}
\end{theorem}
\begin{proof}[Proof idea of Theorem~\ref{thm:FPTAS}]
Consider a forest communication structure $M$. The fact that $G^M$ is a forest implies that given marginal distributions over posteriors $P_1,\ldots,P_k$ for $R_1,\ldots,R_k$, respectively, there exists a signaling scheme $\pi$ inducing $P_1,\ldots,P_k$ if and only if for every $\parentheses*{i_1,i_2}\in S_M$, $P_{i_1}$ is a  mean-preserving spread\footnote{Given two distributions over a compact non-empty set $S\subseteq \mathbb{R}^{d}$ for some $d\geq 1$ with CDFs $F$ and $G$, $G$ is a mean-preserving spread of $F$ if for random variables $X_F\sim F, X_G\sim G$ it holds that: (1)~$\mathbb{E}\brackets*{X_F}=\mathbb{E}\brackets*{X_G}$, i.e.~$X_F,X_G$ have equal means; (2)~$X_G$ has the same distribution as $X_F+Y$, where $Y$ is a random variable with $Y|X_F\equiv 0$ (see, e.g.,~\cite{Mean}).
For $d=1$, this is equivalent to the two distributions having the same mean and $F$ second-order stochastic dominating the spread $G$.} of $P_{i_2}$ and $P_1,\ldots,P_k$ are Bayes-plausible\footnote{A distribution over posteriors $P$ is \emph{Bayes-plausible}~\cite{kamenica2011bayesian} if $\mathbb{E}_{p_{\sigma}\sim P}\brackets*{p_{\sigma}}=p$. By~\cite{Blackwell53,aumann1995repeated}, for a single receiver, a distribution over posteriors is implementable as a signaling scheme if and only if it is Bayes-plausible.} 
(see e.g.,~\citet{brooks2019information}). Given $\epsilon>0$ and assuming, for simplicity, that $\epsilon$ is a reciprocal of an integer, let us discretize the problem. Specifically, let us restrict the marginal posteriors that a signaling scheme is allowed to induce to the grid $W:=\braces*{l\cdot\epsilon}_{0\leq l\leq \frac{1}{\epsilon}}^{\absolute*{\Omega}}$. Define, for every $w\in W, i\in \brackets*{k}$, a variable $x_{w}^{i}$ specifying which probability mass $P_i$ assigns to the posterior $w$. Define further for every $w_1,w_2\in W, i_1,i_2\in\brackets*{k}$ such that $\parentheses*{i_1,i_2}\in S_M$, a variable $y_{w_1,w_2}^{i_1,i_2}$ specifying which probability mass is moved from $w_2$ in $P_{i_2}$ to $w_1$ in $P_{i_1}$. Adding constraints ensuring Bayes-plausibility, equivalence of the probability mass moved from/to each posterior $w$ for each receiver $i$ to $x_{w,i}$ and conservation of the mean results in a polynomial-sized linear program (LP); solving it yields the marginal distributions of an approximately optimal signaling scheme and the signaling scheme can be deduced from a result of~\cite{aumann1995repeated}.
\end{proof}
For the full proof, please see Appendix~\ref{app:FPTAS}. The exact tractability of forest communication structures is an interesting open question.
\subsection{Intractability of General Multi-Channel Persuasion}
\label{sub:hardness}
In this subsection, we introduce \emph{separable supermajority} sender's utility functions, which are obtained (up to some constants) by taking a partition of the set of receivers and summing supermajority functions corresponding to different elements of the partition. We assume that all the receivers are binary-action, with $A_1=\ldots=A_k=\braces*{0,1}$. Moreover, we consider a binary state of nature space $\Omega=\braces*{0,1}$.
\begin{definition}
\label{def:supermaj}
Sender's utility function $u$ is a \emph{separable supermajority} function if there exists a partition $T=\cupdot_{l\in\brackets*{k'}} T_l$ of $\brackets*{k}$ (for some positive integer $k'$), non-negative numbers $r_1,\ldots,r_{k'}\in\mathbb{R}_{\geq 0}$ and non-negative integers $t_1,\ldots,t_{k'}\in\mathbb{Z}_{\geq 0}$ such that $u\parentheses*{p_{\sigma}}=\sum_{l=1}^{k'} r_l\cdot U_{t_l}\parentheses*{p_{\sigma}^i:i\in T_l}$, where $U_{t_l}$ gets $1$ if at least $t_l$ receivers with indices from $T_l$ take action $1$ and $0$ otherwise.\footnote{In other words, each $U_{t_l}$ is a supermajority function for $T_l$.}
\end{definition}
Assume that the sender has an explicit formula for each $U_{t_l}$, and she can evaluate each such function at every point in a constant time. Then finding an optimal signaling scheme for \emph{public} persuasion without externalities is computationally tractable by finding the concavification of (the piecewise-constant function) $u$ and evaluating it at $p_{\sigma}=p$~\cite{kamenica2011bayesian}. Moreover, by~\cite{arieli2019private}, optimal \emph{private} persuasion is computationally tractable for a supermajority sender's utility function; as there are no externalities, considering each partition element separately implies that it is also tractable for separable supermajority functions. We shall prove that it is intractable in the general multi-channel persuasion setting, unless $P=NP$.
\begin{theorem}
\label{thm:hardness}
Consider $\Omega=\braces*{0,1}$ and a separable supermajority sender's utility function $u\parentheses*{\cdot}$ given by an explicit formula. Then, unless $P= NP$, there is no $\poly\parentheses*{k}$-time algorithm that, given a communication structure $M$ as an input, outputs an optimal signaling scheme.
\end{theorem}
\begin{proof}[Proof idea of Theorem~\ref{thm:hardness}]
The proof is based on reducing the minimum $b$-union problem~\cite{vinterbo2002note} to our setting. Specifically, given a universe $\braces*{1,\ldots,w}$ and sets $Q_1,\ldots,Q_t$, we 
consider a multi-channel instance with $t$ communication channels and $w+t$ receivers who are denoted by  $R_1,\ldots,R_w$, $R_{Q_1},\ldots,R_{Q_t}$. Each receiver $R_{Q_j}$ observes a single channel $j\in \brackets*{t}$, while each $R_i$ observes all the channels $j\in \brackets*{t}$ such that $i\in Q_j$. The reduction chooses sender's utility to be a separable supermajority function for the partition $\braces*{\braces*{1},\ldots,\braces*{w},\braces*{Q_1,\ldots,Q_t}}$ that yields extremely high reward when at least $b$ receivers among $R_{Q_1},\ldots,R_{Q_t}$ know for sure that $\omega=1$ but a slight penalty for each receiver among $R_1,\ldots,R_w$ knowing that $\omega=1$.  Such a tradeoff yields that every optimal signaling scheme must fully reveal that $\omega=1$ to exactly $b$ receivers among $R_{Q_1},\ldots,R_{Q_t}$ such that the corresponding $Q_i$s have the smallest possible union.
\end{proof}
For the full proof, please see Appendix~\ref{app:hard}.
\section{Conclusions and Future Work}
\label{sec:conclusions}
In this paper, we characterize a partial ordering over the possible communication structures in Bayesian persuasion according to the sender's preferences that hold for every persuasion instance. We further prove that multi-channel Bayesian persuasion is, generally, computationally intractable even when both public and private persuasion are tractable. The proof proceeds by reduction from the minimum $b$-union problem, which is hard to approximate (as shown for the equivalent maximum $b$-intersection problem in~\cite{xavier2012note}) and -- under an extension of the \emph{dense versus random conjecture for the densest $k$-subgraph problem}~\cite{chlamtac2012everywhere} to hypergraphs -- also hard to approximate to an $O\parentheses*{k^{1/4-\epsilon}}$ multiplicative factor for any $\epsilon>0$~\cite{chlamtavc2017minimizing}. Our reduction does not preserve the approximation ratio, but we conjecture that a similar methodology to ours would allow to prove a hardness of approximation result for multi-channel persuasion. Another interesting future research direction is studying the best signaling schemes of a sender who can affect a given communication structure by adding or removing a limited number of channels.
\section*{Acknowledgements} 
Yakov Babichenko is supported by the Binational Science Foundation BSF grant no.~2018397 and by the German-Israeli Foundation for Scientific Research and Development GIF grant no.~I-2526-407.6/2019. Inbal Talgam-Cohen is a Taub Fellow supported by the Taub Family Foundation and by the Israel Science Foundation ISF grant no.~336/18. 
Haifeng Xu is supported by a Google Faculty Research Award. 
Yakov, Inbal are Haifeng are supported by an NSF-BSF grant (NSF grant no.~CCF-2132506 and BSF grant no.~2021680). Konstantin Zabarnyi is supported by a PBC scholarship for Ph.D. students in data science. The authors thank Fedor Sandomirskiy for his helpful remarks on the connection of our characterization result proof to secret sharing protocols, as well as on improving the presentation of this result. The authors are also grateful to anonymous reviewers for their suggestions on improving the paper.
\bibliographystyle{plainnat}
\bibliography{biblio}
\appendix
\section{Proof That $\aleph$ Signals Suffice to Implement any Signaling Scheme}
\label{app:aleph}
In this appendix, we formulate and prove the following lemma, which is used to prove the characterization Theorem~\ref{thm:order}, but may also be of independent interest.
\begin{lemma}
\label{lem:aleph}
Fix a distribution over states and posteriors $P\in\Delta\parentheses*{\Omega\times\parentheses*{\Delta}^k}$. Then there exists a signaling scheme $\pi:\Omega \to \Delta \parentheses*{\times_{i\in \brackets*{n}} \brackhalf*{0,1}}$ inducing $P$. In particular, $\pi$ uses at most $\aleph$ different signals on each communication channel.
\end{lemma}
\begin{proof}
Fix a signaling scheme $\pi'$ inducing $P$ under $M$. As $\Omega=\brackets*{0,1}$, the CDF of a distribution over $\Omega$ is continuous up to a countable set of points $D$; thus, it is specified by its values on $\mathbb{Q}\cap D$. Therefore, there are only continuum-many possible posteriors. Let us label each possible marginal posterior $p_{\sigma}^i$ by another $q_{\sigma}^i\in \brackhalf*{0,1}^k$ -- it is doable, as $\absolute*{\brackhalf*{0,1}}=\aleph$. For a joint posterior $p_{\sigma}$, denote $q_{\sigma}:=\parentheses*{q_{\sigma}^1,\ldots,q_{\sigma}^k}$.

Suppose that given some value of $\omega$, a certain posterior labeled by $q_{\sigma}$ can be induced by $\pi'$ under $M$. For every $1\leq i\leq k$, take a communication channel $C_{j_i}$ observed by $R_i$ and (possibly) some other receivers $R_{i_1},\ldots,R_{i_{l\parentheses*{i}}}$. The sender should draw $l\parentheses*{i}$ i.i.d.~random variables $e_1,\ldots,e_{l\parentheses*{i}}\sim U\brackhalf*{0,1}$ -- \emph{encryption keys}. Then she should send $q_{\sigma}^i+e_1+\ldots+e_{l\parentheses*{i}}\; \parentheses*{\text{mod } 1}$ via $C_{j_i}$, and send each $e_m$ $\parentheses*{m\in\brackets*{l\parentheses*{i}}}$ via a channel observed by $R_i$, but not by $R_{i_m}$, if such a channel exists; otherwise, send $e_m$ via an arbitrary channel observed by $R_i$. For $j\in\brackets*{n}$, consider $\sigma_j$ to be the concatenation of the signals sent via $C_j$. As $\sigma_j$ is a concatenation of at most $k^2$ values from $\brackhalf*{0,1}$, we can assume w.l.o.g.~that $\sigma_j\in \brackhalf*{0,1}$ for every $j\in\brackets*{n}$ by defining an injective map from $\brackhalf*{0,1}^{k^2}$ to $\brackhalf*{0,1}$.

Fix $1\leq i\leq k$. Note that $R_i$ observes all the elements of the set:
$$
E_i:=\braces*{q_{\sigma}^i+e_1+\ldots+e_{l\parentheses*{i}}\; \parentheses*{\text{mod } 1}, e_1,\ldots, e_{l\parentheses*{i}}}.
$$
Thus, she can decrypt $q_{\sigma}^i$. Moreover, by construction, $R_i$ does not observe some element from each of the sets $E_{i'}$ for every $1\leq i'\leq n$ s.t.~$R_{i'}$ is not information-dominated by $R_i$. Therefore, $R_i$ sees the elements of $E_i$, the elements of some $E_{i_1'},\ldots,E_{i_s'}$ -- where $R_{i_1'},\ldots,R_{i_s'}$ is the list of the receivers information-dominated by $R_i$ -- and some $r^1,\ldots,r^{l'\parentheses*{i}}$, which are uniform in $\brackhalf*{0,1}$ and are independent of each other and of the algebra generated by $E_i$. Therefore:
\begin{align*}
&p_{\sigma}^i\parentheses*{\sigma_j:M_{i,j}=1}=p_{\sigma}^i\parentheses{q_{\sigma}^i+e_1+\ldots+e_{l\parentheses*{i}}\; \parentheses*{\text{mod } 1}, e_1,\ldots, e_{l\parentheses*{i}},q_{\sigma}^{i_1'}+e_{i_1'}+\ldots+e_{l\parentheses*{i_1'}}\; \parentheses*{\text{mod } 1},\\
&e_{i_1'},\ldots, e_{l\parentheses*{i_1'}},\ldots,q_{\sigma}^{i_s'}+e_{i_s'}+\ldots+e_{l\parentheses*{i_s'}}\; \parentheses*{\text{mod } 1}, e_{i_s'},\ldots, e_{l\parentheses*{i_s'}},r^1,\ldots,r^{l'\parentheses*{i}}}=\\
&p_{\sigma}^i\parentheses*{q_{\sigma}^i+e_1+\ldots+e_{l\parentheses*{i}}\; \parentheses*{\text{mod } 1}, e_1,\ldots, e_{l\parentheses*{i}}}=p_{\sigma}^i\parentheses*{q_{\sigma}^i}=q_{\sigma}^i,
\end{align*}
as desired.
\end{proof}

\section{Proof of the Characterization Theorem}
\label{app:characterization}
Similarly to Lemma~\ref{lem:aleph} from Appendix~\ref{app:aleph} proof, let us label each possible marginal posterior $p_{\sigma}^i$ by some $q_{\sigma}^i\in \brackhalf*{0,1}$ ($1\leq i\leq k$) and assume w.l.o.g.~that $S_1,\ldots,S_k\supseteq\brackhalf*{0,1}$. We shall first prove a slight generalization of Corollary~\ref{cor:optimal}; the proof is almost identical to that of Lemma~\ref{lem:aleph}.
\begin{lemma}
\label{lem:tech}
For a communication structure $M$, let $I$ be a subset of receivers no one of which is information-dominated by another receiver among $R_1,\ldots,R_k$. Fix $P\in\Delta\parentheses*{\Omega\times\parentheses*{\Delta\parentheses*{\Omega}}^k}$ induced by some signaling scheme under private persuasion. Let $P_I\in\Delta\parentheses*{\Omega\times\parentheses*{\Delta\parentheses*{\Omega}}^k}$ be a distribution over posteriors that coincides with $P$ on $I$ and coincides with the distribution induced by the no-revelation scheme on $\brackets*{k}\setminus I$. Then $P_I\in\mathcal{P}_{M}$. Furthermore, $P_I$ is induced under $M$ by some signaling scheme with all the signal realizations belonging to $\brackhalf*{0,1}$.
\end{lemma}
\begin{proof}[Proof of Lemma~\ref{lem:tech}]
Fix a signaling scheme $\pi$ inducing $P$ under private persuasion. We shall prove that the restriction of $P$ to $I$ can be induced by a signaling scheme also under $M$, revealing no information to the members of $\brackets*{k}\setminus I$, using signals whose realizations are supported on a subset of $\brackhalf*{0,1}$.

Suppose that given a certain value of $\omega$, a joint posterior labeled by $q_{\sigma}:=\parentheses*{q_{\sigma}^1,\ldots,q_{\sigma}^k}$ can be induced under private persuasion, and the sender wishes to induce $\parentheses*{q_{\sigma}^i:i\in I}$ under $M$. For every $i\in I$, take a communication channel $C_{j_i}$ observed by $R_i$ and (possibly) some other receivers $R_{i_1},\ldots,R_{i_{l\parentheses*{i}}}$. As in Lemma~\ref{lem:aleph} proof, the sender should draw $l\parentheses*{i}$ i.i.d.~random variables $e_1,\ldots,e_{l\parentheses*{i}}\sim U\brackhalf*{0,1}$ that will serve as \emph{encryption keys}, transmit $q_{\sigma}^i+e_1+\ldots+e_{l\parentheses*{i}}\; \parentheses*{\text{mod } 1}$ via $C_{j_i}$, and send each $e_m$ $\parentheses*{m\in\brackets*{l\parentheses*{i}}}$ via a channel observed by $R_i$, but not by $R_{i_m}$ (such a channel exists, as the members of $I$ are not information-dominated by any receivers). Again, as in Lemma~\ref{lem:aleph} proof, upon considering $\sigma_j$ ($j\in\brackets*{n}$) to be the concatenation of the signals on $C_j$, one can assume w.l.o.g.~that $\sigma_j\in \brackhalf*{0,1}$ for every $j\in\brackets*{n}$.

Fix $i\in I$. Note that $R_i$ observes all the elements of the set:
$$
E_i:=\braces*{q_{\sigma}^i+e_1+\ldots+e_{l\parentheses*{i}}\; \parentheses*{\text{mod } 1}, e_1,\ldots, e_{l\parentheses*{i}}}.
$$
Thus, she can decrypt $q_{\sigma}^i$. Moreover, by construction, $R_i$ does not observe some element from each of the sets $E_{i'}$ for $i'\in\ I\setminus\braces*{i}$. That is, $R_i$ sees the elements of $E_i$ and some $s^1,\ldots,s^{l'\parentheses*{i}}$, which are uniform in $\brackhalf*{0,1}$ and are independent of each other and of the algebra generated by $E_i$. Similarly to Lemma~\ref{lem:aleph} proof, we deduce that:
\begin{align*}
    &p_{\sigma}^i\parentheses*{\sigma_j:M_{i,j}=1}=p_{\sigma}^i\parentheses*{q_{\sigma}^i+e_1+\ldots+e_{l\parentheses*{i}}\; \parentheses*{\text{mod } 1}, e_1,\ldots, e_{l\parentheses*{i}},s^1,\ldots,s^{l'\parentheses*{i}}}=q_{\sigma}^i,
\end{align*}
as desired. Similarly, by construction, for every $i\in \brackets*{k}\setminus I$, $R_i$ only observes several i.i.d.~numbers in $\brackhalf*{0,1}$, which induce the same $q_{\sigma}^i$ as the no-revelation signaling scheme, as needed.
\end{proof}
The proof of Theorem~\ref{thm:order} relies on an inductive generalization of the above ideas. The induction base is Corollary~\ref{cor:optimal}, which immediately follows from Lemma~\ref{lem:tech} with $I=\brackets*{k}$. The induction step also uses Lemma~\ref{lem:tech}, together with the following.
\begin{lemma}
\label{lem:tech2}
Fix a communication structure $M$ and $i\in\brackets*{k}$. Let $D_i$ be the set of all the receivers information-dominated by $R_i$. Fix a distribution over states and posteriors $P\in\mathcal{P}_{M}$. Let $P'\in\Delta\parentheses*{\Omega\times\parentheses*{\Delta\parentheses*{\Omega}}^k}$ be a distribution that coincides with $P$, for any value of $\omega$, on $\brackets*{k}\setminus\braces*{i}$, and implies the distribution over $p_{\sigma}^i$ induced when $R_i$ knows the marginal posteriors of the members of $D_i$ as specified by $P$.\footnote{In particular, under $P'$, $R_i$ is not told his own marginal posterior as specified by $P$.} Then there exists a signaling scheme, all signal realizations of which belong to $\brackhalf*{0,1}$, that induces $P'$ under $M$.
\end{lemma}
This lemma allows to extend a given signaling scheme for all the receivers but $R_i$ so that it would include also $R_i$, in a way that reveals to $R_i$ only the information he \emph{must} obtain under \emph{any} extension of the given signaling scheme -- namely, the marginal posteriors of the receivers information-dominated by $R_i$.
\begin{proof}[Proof of Lemma~\ref{lem:tech2}]
The proof is similar to those of Lemma~\ref{lem:aleph} and Lemma~\ref{lem:tech}. Assume that $S_1=\ldots=S_k=\brackhalf*{0,1}$ (it does not make our problem easier, as by Lemma~\ref{lem:aleph} one can assume w.l.o.g.~that the signaling scheme inducing $P$ is supported on at most $\aleph$ different signal realizations). Suppose that under $M$, for some given value of $\omega$, some posterior labeled by $q_{\sigma}$ can be induced, and the sender wishes to induce $\parentheses*{q_{\sigma}^m:m\in \brackets*{k}\setminus\braces*{i}}$ on $\brackets*{k}\setminus\braces*{i}$ and the marginal posterior $p_{\sigma}^i$ obtained when $R_i$ learns $\braces*{p_{\sigma}^m:m\in D_i}$. 

For every $j\in\brackets*{n}$ such that $C_j$ is observed by $R_i$ but by none of $D_i$'s members, let $R_{j_1},\ldots,R_{j_{l\parentheses*{j}}}$ be the receivers observing $C_j$ besides $R_i$.\footnote{As for the other communication channels -- when implementing $P'$, the sender should just send over them the same signal realizations as when implementing $P$, concatenated with the auxiliary signal realizations that we shall introduce below.} The sender should draw $l\parentheses*{j}$ i.i.d.~random variables $e_1,\ldots,e_{l\parentheses*{j}}\sim U\brackhalf*{0,1}$; then she should send $\sigma_j+e_1\;\parentheses*{\text{mod } 1},\ldots,\sigma_j+e_{l\parentheses*{j}}\; \parentheses*{\text{mod } 1}$ via $C_{j}$, and send each $e_m$ $\parentheses*{m\in\brackets*{l\parentheses*{j}}}$ via a channel observed by $R_{j_m}$, but not by $R_{i}$ (such a channel exists, since $\parentheses*{D_i\cup \braces*{R_i}}\cap \braces*{R_{j_1},\ldots,R_{j_{l\parentheses*{j}}}}=\emptyset$). That is -- unlike in  Lemma~\ref{lem:tech} proof, the signal is encrypted with each key \emph{separately} (which allows each of $R_{j_1},\ldots,R_{j_{l\parentheses*{j}}}$ to decrypt it, but not to $R_i$). The rest of the proof is as for Lemma~\ref{lem:tech}, \emph{mutatis mutandis}.
\end{proof}
Now we are ready to prove the characterization result.
\begin{proof}[Proof of Theorem~\ref{thm:order}]
We start with the harder direction and prove by induction on the number of receivers, $k$, that if $S_{M_1}\supseteq S_{M_2}$ then $\mathcal{P}_{M_1}\subseteq\mathcal{P}_{M_2}$, and each distribution $P\in\mathcal{P}_{M_1}$ is implementable under $M_2$ using only signal realizations from $\brackhalf*{0,1}$; it trivially implies $M_1\preceq M_2$.

The claim holds for $k=1$, as by Lemma~\ref{lem:tech}, one can implement persuasion of a single receiver using only signals from $\brackhalf*{0,1}$. Assume that the claim holds for $k\leq m$, for some $m\geq 1$, and consider $k=m+1$. Assume w.l.o.g.~that no two receivers $R_{i_1},R_{i_2}$ ($i_1\neq i_2$) satisfy $\parentheses*{i_1,i_2}, \parentheses*{i_2,i_1}\in S_{M_2}$. This is possible, since one can replace a set of receivers observing exactly the same communication channels by a single receiver. Under this assumption, there exists a receiver $R_i$ who is not information-dominated by any receiver under $M_2$.

Let $P\in\Delta\parentheses*{\Omega\times\parentheses*{\Delta\parentheses*{\Omega}}^k}$ be implementable by a signaling scheme $\pi$ under $M_1$. By the induction hypothesis -- one can implement $P|_{\braces*{R_1,\ldots,R_k}\setminus \braces*{R_i}}$ under $M_2$ using only signal realizations from $\brackhalf*{0,1}$. Let us implement it under $M_2$ by the method ensured from Lemma~\ref{lem:tech2}. Under this implementation, the information transmitted to $R_i$ exactly consists of the marginal posteriors of the receivers information-dominated by $R_i$ under $M_2$. As $S_{M_1}\supseteq S_{M_2}$, our implementation of $P|_{\braces*{R_1,\ldots,R_k}\setminus \braces*{R_i}}$ under $M_2$ never reveals more information to $R_i$ than any of its implementation under $M_1$.

As $R_i$ is not information-dominated by any receiver under $M_2$, by Lemma~\ref{lem:tech} -- the sender can transmit a signal to $R_i$ with all the other receivers learning nothing from it. Hence, the sender can complete implementing $P$ under $M_2$ by sending a signal to $R_i$ consisting of the (encrypted) concatenation of all the signals they should get under $\pi$.

For the second direction, assume that $S_{M_1}\nsupseteq S_{M_2}$ and let us prove that $M_1\npreceq M_2$. Take $\parentheses*{i_1,i_2}\in S_{M_2}\setminus S_{M_1}$. Set $k=2$, a uniform prior $p$ over $\braces*{0,1}\subseteq\Omega$ and $A_1=\ldots=A_k=\braces*{0,1}$. Set for $i\in\brackets*{k}$: $u_i\parentheses*{\omega,a_i=0}=0$, for $i\neq i_2$: $u_i\parentheses*{\omega=1,a_i=1}=1$ and $u_i\parentheses*{\omega=0,a_i=1}=-1$, for $i=i_2$: $u_{i_2}\parentheses*{\omega=1,a_{i_2}=1}=-1$ and $u_{i_2}\parentheses*{\omega=0,a_{i_2}=1}=0$. Assume further that $v\parentheses*{\omega,a_1,\ldots,a_k}=1$ if $a_{i_1}=a_{i_2}=1$ and $v\parentheses*{\omega,a_1,\ldots,a_k}=0$ otherwise. In particular, in our construction, the receivers do not have externalities. Suppose that ties are broken in the sender's favour -- that is, if it is incentive-compatible for $R_{i_1},R_{i_2}$ to take action $1$ -- they do it.

The sender gets a strictly positive expected utility if, with a positive probability, ${p_{\sigma}}^{i_2}\parentheses*{\omega=0}=1$, while ${p_{\sigma}}^{i_1}\parentheses*{\omega=0}\leq \frac{1}{2}$ (otherwise, the sender gets expected utility of $0$). This cannot happen under $M_2$, as $R_{i_1}$ is never less informed than $R_{i_2}$. However, this can happen with probability $\frac{1}{2}$ under $M_1$ -- just always send a fully-informative signal via a communication channel observed by $R_{i_2}$ but not by $R_{i_1}$, and a totally non-informative signal via all the other communication channels. Therefore, in the instance we constructed -- the sender gets a strictly higher expected utility for $M_1$ compared to $M_2$; hence, $M_1\npreceq M_2$, as required.
\end{proof}
\section{Receivers with Externalities}
\label{app:univarsal}
In this section, we show that all the results of Section~\ref{sec:characterization}, including its main Theorem~\ref{thm:order}, continue to hold when receivers' actions may affect each other -- that is, the receivers have \emph{externalities}. Moreover, Lemma~\ref{lem:aleph} from Appendix~\ref{app:aleph} continues to hold.

\vspace{2mm}
\noindent{\bf Universal type space.} After observing the sender's signals $\sigma$, the receivers play a Bayesian game with incomplete information. Each receiver $R_i$ has an infinite sequence of beliefs $b_1^i, b_2^i,\ldots$ that may affect his choice of action in the game. The first-order belief $b_1^i\in\Delta\parentheses*{\Omega}$ is $R_i$'s updated belief about the state of nature, and for $l\geq 2$, 
$$
b_l^i\in \braces*{b_{l-1}^i}\times \Delta\parentheses*{\Omega\times\parentheses*{\times_{i'\in\brackets*{k}\setminus\braces*{ i}}\braces*{\text{all possible $b_{l-1}^{i'}$s}}}}.
$$
In words, the $l$-th order belief of $R_i$ is the concatenation of his $\parentheses*{l-1}$-th order belief and his belief regarding the $\parentheses*{l-1}$-th order beliefs of all other receivers.  To represent these infinite levels of beliefs
we use Harsanyi's \emph{universal type space}~\cite{harsanyi1967games} as formulated by~\cite{mertens1985formulation,brandenburger1993hierarchies}. This formulation constructs $k$ universal type spaces $T_1,\ldots,T_k$ for the $k$ receivers, replacing infinite belief hierarchies with well-defined infinite-dimensional type spaces.\footnote{In particular, for every receiver $R_i$ there exists a homeomorphism from $T_i$ to $\Delta\parentheses*{\Omega\times \parentheses*{\times_{i'\in\brackets*{k}\setminus\braces*{i}}T_{i'}}}$. Therefore, $T_1,\ldots,T_k$ capture the whole belief hierarchy $\braces*{b_1^i, b_2^i,\ldots}_{i\in\brackets*{k}}$ of the receivers.} Define $T:=T_1\times\ldots\times T_k$, and
let $t_i$ be a random variable representing the actual \emph{type} of receiver $R_i$. Note that $t_i=t_i\parentheses*{\sigma_j:M_{i,j}=1}$, i.e., receiver $R_i$'s type depends on the (random) signals over channels he observes.

As $\Omega=\brackets*{0,1}$, and the belief hierarchy captured by a universal type space is of countable depth -- each of $T_1,...,T_k$ is countable. Therefore, all the results of Section~\ref{sec:characterization}, as well as Lemma~\ref{lem:aleph}, continue to hold, with the following changes:
\begin{enumerate}
    \item The common prior $p$ is an element of $\Delta\parentheses*{\Omega\times T}$ rather than of $\Delta\parentheses*{\Omega}$ -- that is, it is a distribution over the states of nature and the receivers' types.
    \item The marginal posterior $p_{\sigma}^i$ ($1\leq i\leq k$) equals $t_i \in T_i$, the type of $R_i$.
\end{enumerate}
\section{Proof of the Positive Computational Result}
\label{app:FPTAS}
\begin{proof}[Proof of Theorem~\ref{thm:FPTAS}]
Let $P_1,\ldots,P_k:\Omega\to\brackets*{0,1}$ be the marginal distributions over posteriors of $R_1,\ldots, R_k$, respectively. By~\cite{kamenica2011bayesian}, one can assume that each of $P_1,\ldots,P_k$ is Bayes-plausible: $\mathbb{E}_{p_{\sigma}\sim P}\brackets*{p_{\sigma}}=p$; otherwise, no signaling scheme can induce them. Since $G^M$ is a directed forest, we get from~\cite{brooks2019information} that there exists a signaling scheme $\pi$ that induces the marginal distributions $P_1,\ldots,P_k$ on $R_1,\ldots,R_K$, respectively, if and only if for every $\parentheses*{i_1,i_2}\in S_M$, the distribution $P_{i_1}$ is a mean-preserving spread of $P_{i_2}$.

Fix $0<\epsilon<1$. For simplicity -- assume w.l.o.g.~that $\epsilon$ is a reciprocal of a positive integer.\footnote{For every $0<\epsilon<1$, there exists a positive integer $d$ such that $\frac{\epsilon}{2}\leq \frac{1}{d}\leq\epsilon$; then $\poly\parentheses*{k,\frac{1}{\epsilon}}=\poly\parentheses*{k,\frac{1}{1/d}}$.} Define the grid $W:=\braces*{l\cdot\epsilon}_{0\leq l\leq \frac{1}{\epsilon}}^{\absolute*{\Omega}}$. We shall discretize the problem by restricting the admissible marginal posterior distributions to $W$.

Introduce a variable $x_{l_1,\ldots,l_{\absolute*{\Omega}}}^i$ for every $0\leq l_1,\ldots,l_{\absolute*{\Omega}}\leq\frac{1}{\epsilon}$ and $i\in\brackets*{k}$ to specify which probability mass $P_i$ assigns to $\parentheses*{{l_1\epsilon,\ldots,l_{\absolute*{\Omega}}}\epsilon}\in W$. Furthermore, introduce a variable $x_{l_1,\ldots,l_{\absolute*{\Omega}},r_1,\ldots,r_{\absolute*{\Omega}}}^{i_1,i_2}$ for every $0\leq l_1,\ldots,l_{\absolute*{\Omega}},r_1,\ldots,r_{\absolute*{\Omega}}\leq\frac{1}{\epsilon}$ and $i_1,i_2\in\brackets*{k}$ such that $\parentheses*{i_1,i_2}\in S_M$ to specify which probability mass is moved from $\parentheses*{{r_1\epsilon,\ldots,r_{\absolute*{\Omega}}}\epsilon}\in W$ to $\parentheses*{{l_1\epsilon,\ldots,l_{\absolute*{\Omega}}}\epsilon}\in W$ when treating $P_{i_1}$ as a mean-preserving spread of $P_{i_2}$. Consider the following LP:
\begin{align*}
    \text{max}\;\;\;\;&\sum_{0\leq l_1,\ldots,l_{\absolute*{\Omega}} \leq \frac{1}{\epsilon}}\sum_{i\in \brackets*{k}} u^{i}\parentheses*{l_1\epsilon,\ldots,l_{\absolute*{\Omega}}\epsilon} \cdot x_{l_1,\ldots,l_{\absolute*{\Omega}}}^{i}\\
    \text{s.t.}\;\;\;\;\;\;&\sum_{0\leq l_1,\ldots, l_{b-1},l_{b+1},l_{\absolute*{\Omega}}\leq \frac{1}{\epsilon}} l_{b}\epsilon\cdot x_{l_1,\ldots,l_{\absolute*{\Omega}}}^{i}=p^b\;\;\;\;\;\;\;\;\;\;\;\;\;\;\;\; b\in\brackets*{\absolute*{\Omega}},\;\; i\in\brackets*{k}\\
    &\sum_{0\leq r_1,\ldots,r_{\absolute*{\Omega}}\leq\frac{1}{\epsilon}} y_{l_1,\ldots,l_{\absolute*{\Omega}},r_1,\ldots,r_{\absolute*{\Omega}}}^{i_1,i_2}=x_{l_1,\ldots,l_{\absolute*{\Omega}}}^{i_1}\;\;\;\;\;\;\;\;\;\;\;\;0\leq l_1,\ldots,l_{\absolute*{\Omega}}\leq\frac{1}{\epsilon},\;\;\parentheses*{i_1,i_2}\in S_M\\
    &\sum_{0\leq l_1,\ldots,l_{\absolute*{\Omega}}\leq\frac{1}{\epsilon}} y_{l_1,\ldots,l_{\absolute*{\Omega}},r_1,\ldots,r_{\absolute*{\Omega}}}^{i_1,i_2}=x_{r_1,\ldots,r_{\absolute*{\Omega}}}^{i_2}\;\;\;\;\;\;\;\;\;\;\;\;0\leq r_1,\ldots,r_{\absolute*{\Omega}}\leq\frac{1}{\epsilon},\;\;\parentheses*{i_1,i_2}\in S_M\\
    &\sum_{0\leq l_1,\ldots, l_{b-1},l_{b+1},l_{\absolute*{\Omega}}\leq\frac{1}{\epsilon}} l_{b} y_{l_1,\ldots,l_{\absolute*{\Omega}},r_1,\ldots,r_{\absolute*{\Omega}}}^{i_1,i_2}=r_b\;\;\;\;\;\;\;\;\;\;b\in\brackets*{\absolute*{\Omega}},\;\;0\leq r_1,\ldots,r_{\absolute*{\Omega}}\leq\frac{1}{\epsilon},\;\;\parentheses*{i_1,i_2}\in S_M\\
    &\sum_{0\leq l_1,\ldots,l_{\absolute*{\Omega}}\leq\frac{1}{\epsilon}} x_{l_1,\ldots,l_{\absolute*{\Omega}}}^{i}=1\;\;\;\;\;\;\;\;\;\;\;\;\;\;\;\;\;\;\;\;\;\;\;\;\;\;\;\;\;\;\;\;\;\;\;\;\; i\in\brackets*{k}\\
    &x_{l_1,\ldots,l_{\absolute*{\Omega}}}^{i}\geq 0 \;\;\;\;\;\;\;\;\;\;\;\;\;\;\;\;\;\;\;\;\;\;\;\;\;\;\;\;\;\;\;\;\;\;\;\;\;\;\;\;\;\;\;\;\;\;\;\;\;\;\;\;\;\;\;\; 0\leq l_1,\ldots,l_{\absolute*{\Omega}}\leq\frac{1}{\epsilon},\;\; i\in\brackets*{k}\\
    &y_{l_1,\ldots,l_{\absolute*{\Omega}},r_1,\ldots,r_{\absolute*{\Omega}}}^{i_1,i_2}\geq 0 \;\;\;\;\;\;\;\;\;\;\;\;\;\;\;\;\;\;\;\;\;\;\;\;\;\;\;\;\;\;\;\;\;\;\;\;\;\;\;\;\;\;\;\; 0\leq l_1,\ldots,l_{\absolute*{\Omega}},r_1,\ldots,r_{\absolute*{\Omega}}\leq\frac{1}{\epsilon},\;\; \parentheses*{i_1,i_2}\in S_M.
    \end{align*}
The first constraint family ensures Bayes-plausibility, the second, third, fourth and seventh -- that $P_{i_1}$ is a mean-preserving spread of $P_{i_2}$ for every $\parentheses*{i_1,i_2}\in S_M$, the fifth and sixth -- that $P_1,\ldots,P_k$ are probability distributions. Therefore, the solutions of the above LP correspond exactly to $P_1,\ldots,P_k$ that can be induced by a signaling scheme, provided that $\supp\parentheses*{P_1},\ldots,\supp\parentheses*{P_k}\subseteq W$. Moreover, the LP has $k\cdot \parentheses*{1+\frac{1}{\epsilon}}^{\absolute*{\Omega}}+\absolute*{S_M}\cdot \parentheses*{1+\frac{1}{\epsilon}}^{2\absolute*{\Omega}}=\poly\parentheses*{k,\frac{1}{\epsilon}}$ variables since $\absolute*{\Omega}=O\parentheses*{1}$, and $\absolute*{\Omega}\cdot k+ \parentheses*{2+\absolute*{\Omega}}\cdot\parentheses*{1+\frac{1}{\epsilon}}\cdot\absolute*{S_M}+k+\parentheses*{1+\frac{1}{\epsilon}}\cdot k+\parentheses*{1+\frac{1}{\epsilon}}^2\cdot \absolute*{S_M} =\poly\parentheses*{k,\frac{1}{\epsilon}}$ constraints; thus, it is solvable in $\poly\parentheses*{k,\frac{1}{\epsilon}}$-time.

Let $\pi^*$ be an optimal signaling scheme and let $P^{*}_1,\ldots,P^{*}_k$ be the corresponding marginal distributions over posteriors. Consider now the following discretization procedure. Initialize $P_1=P^{*}_1$,\ldots, $P_k=P^{*}_k$ and traverse the forest $G^M$ in preorder. When visiting a node representing a certain receiver $R_i$ -- move all the probability mass in $P_i$ contained in a certain grid hypercube to an extreme point of this hypercube for which the value of $u^i$ is maximal. Moreover, when moving a mass from some marginal posterior $q$ to another marginal posterior $\tilde{q}$ -- for any descendant $R_{i'}$ of $R_i$ in $G^M$, move the same mass by the same vector from every $q'\in\supp\parentheses*{P_{i'}}$ such that a mass is moved from $q'$ to $q$ when considering $P_i$ as a mean-preserving spread of $P_{i'}$.\footnote{Note that some probability mass may be moved outside $\Delta\parentheses*{\Omega}$, but this issue will be fixed on a later iteration.} This transformation preserves the mean-preserving spread constraints between $R_i$ and its descendants in $G^M$, and cause an additive loss of $O\parentheses*{\epsilon}$ to the expectation of $u\parentheses*{\cdot}$.\footnote{Here we use that $u^1,\ldots,u^k$ are upper semi-continuous and either $O\parentheses*{1}$-Lipschitz or piecewise constant; the claim holds for small enough values of $\epsilon$ for which every grid hypercube is covered by a single piece of every piecewise constant function among $u^1,\ldots,u^k$.}. Such an operation might violate the Bayes-plausibility constraints and the mean-preserving spread constraints between $R_i$ and its ancestors in $G^M$, which should be fixed by moving $O\parentheses*{\epsilon}$ fractions of the probability mass from all of the marginal distributions over posteriors of the ancestors of $R_i$ to the extreme points of $\Delta\parentheses*{\Omega}$. The resultant expected loss for $u\parentheses*{\cdot}$ is $\Delta\parentheses*{\Omega}$, as $u^1,\ldots,u^k$ are $O\parentheses*{1}$-bounded (since they are upper semi-continuous).

The above process yields a feasible solution to our LP, consisting of the marginal distributions of a signaling scheme together with the variables $\braces*{y_{l_1,\ldots,l_{\absolute*{\Omega}},r_1,\ldots,r_{\absolute*{\Omega}}}^{i_1,i_2}}_{0\leq l_1,\ldots,l_{\absolute*{\Omega}},r_1,\ldots,r_{\absolute*{\Omega}}\leq\frac{1}{\epsilon},\; \parentheses*{i_1,i_2}\in S_M}$ representing the belief distributions of Blackwell-ordered pairs of receivers in terms of mean-preserving spreads. By~\cite{aumann1995repeated}, given such a representation -- one can compute the corresponding signaling scheme (that is, a table specifying, for every state of nature, the probabilities for sending the possible signal realizations) via an explicit formula in a time polynomial in the number of variables. As the additive loss during the procedure in the expectation of the sender's utility is $O\parentheses*{\epsilon}$, the proof is complete.
\end{proof}
\section{Proof of the Negative Computational Result}
\label{app:hard}
\begin{proof}[Proof of Theorem~\ref{thm:hardness}]
We shall reduce the minimum $b$-union problem, which is known to be $NP$-hard~\cite{vinterbo2002note}, to the problem of finding an optimal signaling scheme under a communication structure $M$ given as an input. Consider a problem instance with a fixed universal set $\braces*{1,\ldots,w}$ and input sets $Q_1,\ldots,Q_t\subseteq \braces*{1,\ldots,w}$ such that the goal is to find $1\leq i_1<\ldots<i_{b}\leq t$ for which $\absolute*{Q_{i_1}\cup\ldots\cup Q_{i_{b}}}$ is minimal, where $b$ is given as an input.

Define a multi-channel Bayesian persuasion setting instance as follows. Take a uniform prior distribution $p=\parentheses*{\frac{1}{2},\frac{1}{2}}$ over $\Omega=\braces*{0,1}$. Introduce receivers $R_{Q_1},\ldots,R_{Q_t},R_1,\ldots,R_w$. Take such a communication structure $M$ that each one of $R_{Q_1},\ldots,R_{Q_t}$ sees a separate communication channel, and each one of $R_1,\ldots,R_w$ sees exactly the channels observed by $R_{Q_{i'}}$ for which $i\in Q_{i'}$. Under this communication structure, $S_M=\braces*{\parentheses*{i_1,Q_{i_2}}:i_1\in Q_{i_2}}$.

Define the partition $T=\braces*{\braces*{R_1},\ldots,\braces*{R_w},\braces*{R_{Q_1},\ldots,R_{Q_t}}}$ of the receivers' set. Set $r_1=\ldots=r_w=1$, $r_{Q_1,\ldots,Q_t}=4w$ in Definition~\ref{def:supermaj} to specify the sender's separable supermajority utility function.

Let $U_{R_l}$ ($l\in\brackets*{w}$) be $0$ if $p_{\sigma}^{l}\parentheses*{\omega=1}>0.9$ and $1$ otherwise. Moreover, let $U_{Q_1,\ldots,Q_t}$ be $1$ if at least $b$ among $R_{Q_1},\ldots,R_{Q_t}$ have a posterior probability of $1$ for $\braces*{\omega=1}$ and $0$ otherwise. Note that the above functions are upper semi-continuous and piecewise constant with $2$ pieces; thus, there exists a direct sender's utility function inducing $u$ for binary-action receivers with ties broken in sender's favour~\cite{kamenica2011bayesian}.

Let $h$ be the minimal cardinality of a union of $b$ sets among $Q_1,\ldots,Q_t$. Then no signaling scheme can achieve expected sender's utility greater than $p\parentheses*{\omega=0}\cdot\parentheses*{1\cdot w}+p\parentheses*{\omega=1}\cdot\parentheses*{4w-h}=\frac{5w-h}{2}$. Moreover, every signaling scheme achieving this bound truthfully reveals the state to $b$ receivers among $R_{Q_1},\ldots,R_{Q_t}$ representing the sets constituting a minimum-cardinality union of $b$ sets among $Q_1,\ldots,Q_t$, and does so in a way that no receiver $R_{i}$ for $i\in\brackets*{w}$ not in the above set union has, with a positive probability, $p_{\sigma}^i\parentheses*{\omega=1}>0.9$. Note that at least one such scheme exists -- the sender should transmit a fully informative signal via the each of the channels observed exclusively by one of the $b$ receivers among $R_{Q_1},\ldots,R_{Q_t}$ corresponding to $b$ sets with a minimum-cardinality union, and transmit a totally non-informative signal via all the other channels. Therefore, the $b$-union problem reduces to finding an optimal signaling scheme in multi-channel persuasion with a separable supermajority sender's utility function, as desired.
\end{proof}
\end{document}